\documentclass[12pt]{article}

\usepackage[margin=1in]{geometry}

\usepackage{graphicx}
\usepackage[unicode]{hyperref}
\usepackage{amsthm}
\usepackage{amssymb}
\usepackage{amsmath}
\usepackage{mathrsfs}
\usepackage{verbatim}
\usepackage{algorithm}
\usepackage{mathtools}
\usepackage{xspace}
\usepackage{enumitem}
\usepackage{color}
\usepackage{url}

\newtheorem{lemma}{Lemma}
\newtheorem{defi}{Definition}

\newtheorem{theorem}{Theorem}

\newcommand{\vcell}{\mathrm{vcell}}
\newcommand{\disc}{\Delta} 
\newcommand{\hval}{h}     
\newcommand{\wt}{\widetilde}
\newcommand{\netpoint}[1]{\mathrm{N}(#1)}

\newcommand{\tree}{T}
\newcommand{\R}{\mathbb{R}}
\newcommand{\rep}{\mathrm{rep}}
\newcommand{\diam}{\mathrm{diam}}
\newcommand{\trep}{t} 
\newcommand{\eps}{\varepsilon}
\newcommand{\Cech}{\v{C}ech\xspace}
\newcommand{\dd}{\Delta}
\newcommand{\td}{\Delta_{\trep}}
\newcommand{\tds}[1]{\Delta_{#1}}
\newcommand{\ball}{\mathbb{B}}
\newcommand{\meb}{\mathrm{meb}}
\newcommand{\mebrad}{\mathrm{rad}}

\newcommand{\parent}{\mathrm{parent}}

\newcommand{\Rel}{\mathrm{Rel}}
\newcommand{\kapa}{k}
\newcommand{\level}{\ell}
\newcommand{\Approx}{\mathcal{A}}
\newcommand{\distance}[2]{\|#1-#2\|}
\newcommand{\nodedist}[2]{\mathrm{dist}(#1,#2)}

\DeclarePairedDelimiter{\ceil}{\lceil}{\rceil}
\DeclarePairedDelimiter{\floor}{\lfloor}{\rfloor}

\begin{document}

\title{Local Doubling Dimension of Point Sets}
\author{Aruni Choudhary\footnote{Max Planck Institute for Informatics, Saarbr\"ucken, Germany
 (aruni.choudhary@mpi-inf.mpg.de)} \and Michael Kerber\footnote{Max Planck Institute for Informatics, 
Saarbr\"ucken, Germany (mkerber@mpi-inf.mpg.de)}}

\date{\today}
\maketitle

\begin{abstract}

We introduce the notion of $\trep$-restricted doubling dimension of a point set in Euclidean space as the local
intrinsic dimension up to scale $\trep$. 
In many applications information is only relevant for a fixed range of scales.
We present an algorithm to construct a hierarchical net-tree up to scale $\trep$ which we denote as the 
net-forest. We present a method based on \emph{Locality Sensitive Hashing}
to compute all near neighbours of points within a certain distance. 
Our construction of the net-forest is probabilistic, and we guarantee that with high probability, 
the net-forest is supplemented with the correct neighbouring information. 
We apply our net-forest construction scheme to create an approximate \Cech complex up to a fixed scale;
and its complexity depends on the local intrinsic dimension up to that scale.

\end{abstract}

\section{Introduction}

\paragraph{Motivation}

Often, one wants to perform tasks on data which lives in high dimensional spaces. 
Typically, algorithms for manipulating such high dimensional data take
exponential time with respect to the ambient dimension. 
This is frequently quoted as the ``curse of dimensionality''.
\mbox{In many} cases, however, practical input instances lie on low dimensional manifolds and 
a natural question arises as to how do we exploit this structural property for
computationally feasible algorithms.

A well-established approach is to define a special notion of dimension on a point set.
The \emph{doubling dimension} of a point set $P$ is the smallest integer $\dd$ such that every ball
centered at $p\in P$ of radius $R$ is covered by at most $2^{\dd}$ non-empty balls of radius $R/2$ for any
$R$. For instance, if $P$ is a sample of an affine subspace of dimension $k$, it holds that $\dd=k$,
and often, $\dd\ll d$ holds for more general samples of $k$-manifolds. 
A common goal is therefore to replace the exponential dependency on $d$ by $\dd$ in the complexity
of geometric algorithms.

The concept of \emph{(hierarchical) net-trees} can be seen as a generalization of quadtrees
and allows for the translation of quadtree-based algorithms (which are exponential in $d$)
to cases with small $\dd$. Technically, a net-tree provides a hierarchy of \emph{nets}
which summarize the point set in terms of a clustering scheme on different scales.
For $n$ points with doubling dimension $\dd$, a net-tree
can be constructed in expected $2^{O(\dd)}O(n\log n)$ time, matching the time for constructing
a quadtree except for replacing $d$ with $\dd$~\cite{hm-fast}.
As an application of particular importance, net-trees permit 
the efficient construction of \emph{well-separated pair decomposition} (WSPD) which have various
applications in geometric approximation, such as constructing \emph{spanners}, finding \emph{approximate
nearest-neighbours}, approximating the \emph{diameter} and the \emph{closest-pair distance}.

In some applications, it is natural to upper bound the range of scales under investigation.
In such cases, the doubling dimension does not capture the intrinsic complexity of the problem at hand,
since it may be caused by a ball that is beyond the range of considered scales.
Moreover, the net-tree construction of~\cite{hm-fast} proceeds in a top-down fashion, 
considering the high scales of the point set first. It therefore suffers from potentially bad
large-scale properties of the point set, even when these properties are irrelevant for
the given application.

\paragraph{Contributions}
In this paper, we introduce the concept of $\trep$-restricted doubling dimension $\td$, 
which is the smallest integer such that any ball centered at $p\in P$ of radius $R\le\trep$
is covered by at most $2^{\td}$ non-empty balls of radius $R/2$.
For simplicity of presentation, we restrict ourselves to the case of point sets in Euclidean space
and postpone a more general treatment to an extended version of the paper.
We present an algorithm to construct a \emph{net-forest}, 
which contains the relevant data of a net-tree up to scale $\trep$. 
The runtime of the construction depends on $\tds{C\trep}$ where $C$ is a value independent of $n$
and is defined later on.
We hence remove the dependence on the doubling dimension $\dd$. 
The major geometric primitive of our algorithm is to find all neighbours of a point \mbox{$p\in P$}
with a distance of at most $\Theta(\trep)$.
We propose an approach based on \emph{Locality Sensitive Hashing} (LSH)
from~\cite{lsh}. We have a trade-off between the size of $C$ and the exponent of $n$
in the complexity bound.
The LSH based construction of the net-forest yields an expected runtime of
$O\left( dn^{1+\rho}\log n(\log n + (14/\rho)^{\tds{7t/\rho}}) \right)$
where $\rho\in(0,1)$ is a parameter which can be chosen to be as small as desired. 
Comparing this bound with the full net-tree construction,
our approach makes sense if $n^{\rho}\log n \ll  2^{O(\dd)}$, that is, $\dd$ is sufficiently large and
$\tds{O(t)}\ll \dd$.

As a consequence of our result, we can construct the part of the WSPD where all pairs
are in distance at most $\Theta(\trep)$, adapting the construction scheme of~\cite[Sec.5]{hm-fast}.
That means that any application of WSPD that restricts its attention to low scales can profit from
our approach.

As a further application, we show how to approximate \Cech complexes using net-forests;
\Cech complexes are a standard tool for capturing topological properties of a point cloud.
Such a complex depends on a scale parameter;
in particular, in the context of persistent homology~\cite{eh-computational}, 
\Cech filtrations are considered, which encode \Cech complexes at various scales.
In~\cite{cks-approximate}, an approximate filtration of size
$n(\frac{2}{\eps})^{O(\kapa\cdot\dd)}$ has been constructed using net-trees.
``Approximate'' means that the two filtrations are interleaved in the sense of~\cite{chazal}
and therefore yield similar persistence diagrams.
However, because of the large size of filtrations, it is common to limit their construction 
to an upper threshold value $\trep$. With our results, we can construct such a upper-bound filtration
of size $n(\frac{2}{\eps})^{O(\kapa\cdot \tds{O(\trep)})}$, thus replacing the doubling dimension 
in~\cite{cks-approximate} by the $O(\trep)$-restricted doubling dimension.

\paragraph{Organization of the paper}
Section~\ref{sec:background} gives a brief overview of doubling spaces and net-trees. We introduce the
concept of the restricted version of the doubling dimension in Section~\ref{sec:tdoubling}. 
In Section~\ref{sec:net-forest} we present an algorithm to construct the net-forest up to a certain scale. 
Our algorithm uses the concept of LSH which we detail in Section~\ref{sec:near-neighbour}. 
In Section~\ref{sec:apps} we give an overview of WSSDs and adapt their construction to use the net-forest.
We summarize our results and conclude in Section~\ref{sec:conclusion}.

\section{Background}
\label{sec:background}

We fix $P$ to be a finite point set consisting of $n$ points throughout.
As mentioned before, we restrict our attention to the Euclidean case $P\subset\R^d$, 
although some of the presented
concepts could be extended to arbitrary metric spaces with some additional effort. 
In particular, the distance between any two points can be computed in $O(d)$ time 
for Euclidean setups.

\paragraph{Doubling dimension}
A \emph{discrete ball} centered at a point $q \in P$ with radius $r$ is the set of points $Q\subseteq P$ which 
satisfy $\distance{p}{q}\leq r$ for all $q\in Q$.
The \emph{doubling constant}~\cite{asd-dd,talwar} 
is the smallest integer $\lambda$ such for all $p\in P$ and all $r>0$, the discrete ball centered at $p$ of radius
$r$ is covered by $\lambda$ discrete balls of radius $r/2$.
The \emph{doubling dimension} $\dd$ of $P$ is  $\ceil*{\log_2 \lambda}$. 
For example, a point set that is sampled from a $k$-dimensional subspace
has a doubling dimension of $k$, independent of the ambient dimension $d$.
In contrast, the $d$ boundary points of the standard $(d-1)$-simplex form a doubling space of dimension $\ceil*{\log_2 d}$.
Even worse, we can construct a subset of doubling dimension $\Theta(d)$ by placing $2^{\Theta(d)}$ points inside 
the unit ball
in $\R^d$ such that any two points have a distance of at least $3/2$ (the existence of such a point set follows by
a simple volume argument).
It is NP-hard to calculate the doubling dimension of a metric \cite{lee-dstar} but it can be approximated within
a constant factor \cite[Sec.9]{hm-fast}. 

\paragraph{Nets and Net-trees}
A subset $Q \subseteq P$ is an \emph{$(\alpha,\beta)$-net}, denoted by $\mathcal{N}_{\alpha,\beta}$, if
all points in $P$ are in distance at most $\alpha$ from some point in $Q$
and the distance between any two points in $Q$ is \mbox{at least $\beta$}.
Usually, $\alpha$ and $\beta$ are coupled, that is, $\beta=\Theta(\alpha)$, in which case we talk about a 
\emph{net at scale $\alpha$}.

We can represent a nested sequence of nets for increasing scales $\alpha$ using a rooted tree structure, 
called the \emph{net-tree}~\cite{hm-fast}. 
It has $n$ leaves, each representing a point of $P$, and each internal node has at least two children. 
Every tree-node $v$ represents the subsets of points given by the sub-tree rooted at $v$; we denote this set by $P_v$.
Every $v$ has a representative, $\rep_v\in P_v$ that equals the representative of one of its children if $v$ is not a leaf. 
Moreover, $v$ is associated with an integer $\level(v)$ called the \emph{level} of v which satisfies 
$\level(v) < \level(\parent(v))$, where $\parent(v)$ is the parent of $v$ in the tree. 
Finally, each node satisfies the following properties

 \begin{itemize}
  \item \emph{Covering property}: $P_v \subseteq \ball(\rep_v,\frac{2\tau}{\tau-1}\cdot\tau^{\level(v)})$ 
  \item \emph{Packing property}: 
  $P_v \supseteq P \bigcap \ball(\rep_v,\frac{\tau-5}{2\tau(\tau-1)}\cdot\tau^{\level(\parent(v))})$
 \end{itemize}
where $\ball(p,r)$ denotes the ball centered at $p$ with radius $r$ and $\tau=11$.

The covering and packing properties ensure that each node $v$ has at most $\lambda^{O(1)}$ 
children where $\lambda$ is the doubling constant for $P$. Moreover, for any $\alpha$,
a net at scale $\alpha$ can be read off from the net-tree immediately;
see~\cite[Prop.2.2]{hm-fast} for details.
A net-tree can be constructed deterministically in time $2^{O(\dd)}O(dn \log (n\cdot \Phi))$ where $\Phi$ 
represents the \emph{spread} of $P$,
using the greedy clustering scheme of Gonzalez \cite{gon} as a precursor to the tree construction. 
The dependence on spread can be eliminated by constructing the tree in $2^{O(\dd)}dn \log n$ time in expectation
(the additional factor of $d$ compared to~\cite{hm-fast} accounts for the fact that we fixed the Euclidean metric,
and therefore take into account the cost of computing distances in our computational model).
The net-tree construction is oblivious to knowing the value of $\dd$.
One can extract a net at scale $\ell$~\cite[Pro.2.2]{hm-fast} by collecting the set of nodes from $\tree$ 
satisfying the condition $\mathcal{N}(\ell)=\{\rep_v|\level(v) < \ell \le \level(\parent(v)\}$.
The net-tree can be augmented to maintain, for each node $u$, a list of close-by nodes with similar diameter.
Specifically, for each node $u$ the data structure maintains the set
\begin{equation}
\label{eqn:rel}
 \Rel(u) := \{v \in N\ |\ \level(v)\le \level(u) < \level(\parent(v))\  and\ \distance{\rep_u}{\rep_v} 
 \le 14\tau^{\level(u)}\}.
\end{equation}
$\Rel(.)$ is computed during the construction without additional cost.

\section{$\trep$-restricted doubling dimension}
\label{sec:tdoubling}

\begin{defi}
The \emph{$\trep$-restricted doubling constant} of $P$ is the smallest
positive integer $\lambda_{\trep}$ such that all the points in any discrete ball centered at $p\in P$ of radius $r$ with 
$r\le \trep$ are covered by $\lambda_{\trep}$ non empty balls of radius $r/2$. The corresponding 
\emph{$\trep$-restricted doubling dimension} $\td$ is $\ceil*{\log \lambda_{\trep}}$.
\end{defi}

By definition, $\td \le \dd$ for any $P$. More precisely, $\td$ is zero for $\trep$ smaller than the closest-pair
distance of $P$, and equals $\dd$ when $\trep$ is the diameter of $P$.
While the doubling dimension of samples from an affine subspace of dimension $k$ is bounded by $k$, 
this is not generally true for samples of $k$-manifolds where $\dd$ increases due to curvature. 
To sketch an extreme example, consider an \mbox{\emph{almost space-filling}} curve $\gamma$ in $\R^d$ which has distance 
at most $\eps$ to any point of the unit ball, where $\eps$ is chosen small enough. 
We let $P$ be a sufficiently dense sample of $\gamma$. 
While $\td=1$ for small values of $\trep$, we claim that $\td=\Theta(d)$ for $\td=1$; 
indeed, any sparser covering of the unit ball
with balls of radius $1/2$ would leave some portion of the ball uncovered, and by construction, 
$\gamma$ goes through that uncovered region, so that some point in $P$ is missed. 
We skip a more formal treatment of this argument.

The ``badness'' of the previous example stems from the difference between Euclidean and geodesic distance of points 
lying on a lower-dimensional manifold. 
A common technique for approximating the geodesic distance is through the \emph{shortest-path metric}:
Let $G=(P,E)$ denote the graph whose
edges are defined by the pairs of points of Euclidean distance at most $\trep$. The distance of two points $p$ and $q$
is then defined as the length of the shortest path from $p$ to $q$ (we assume for simplicity that $G$ is connected).
The concept of doubling dimensions extends to any metric space and we let $\dd'$ denote the doubling dimension
of $P$ equipped with the shortest path metric. 
While $\td$ and $\dd'$ appear to be related, $\dd'$ can be much larger than $\td$ in general.
Moreover, using the shortest-path metric raises the question of how to compute shortest path distances
efficiently, if the cost of metric queries is taken into account.

\section{Net-forests}
\label{sec:net-forest}

We next define an appropriate data structure for point sets of small $\trep$-restricted doubling dimension,
where $\trep$ is a parameter of the construction.
Informally, a \emph{net-forest} is the subset of a net-tree obtained by truncating all nodes
above scale $\trep$. More precisely, it is a collection of net-trees with roots $v_1,\ldots,v_k$ such that
the representatives $\rep_{v_1},\ldots,\rep_{v_k}$ form a $(\trep,\trep)$-net and the point sets
$P_{v_1},\ldots,P_{v_k}$ are disjoint and their union covers~$P$.
We define $\Rel(u)$ for a node in the forest the same way as in \eqref{eqn:rel} as the set of net-forest nodes
that are close to $u$ and of similar scale.
As for net-trees, we call a net-forest \emph{augmented} if each node $u$ is equipped with $\Rel(u)$.

\paragraph{Construction} 
Our algorithm for constructing a net-forest is a simple adaption of the net-tree algorithm: we construct 
a $(\trep,\trep)$-net of $P$ by clustering the point set and assign each point in $P$ to its closest net-point.
Each root in the net-forest represents one of the clusters. We also compute $\Rel(u)$ for each root
by finding the close-by clusters to $u$. Having this information, we can simply run the net-tree algorithm 
from~\cite{hm-fast} individually on each cluster to construct the net-forest. 
For augmenting it, we use the top-down traversal strategy
as described in~\cite[Sec.3.4]{hm-fast}, inferring the neighbours of a node from the neighbours
of its parent~--since we have set up $\Rel(\cdot)$ for the roots of the forest, this
strategy is guaranteed to detect neighbours even if they belong to different trees of the forest.

Both the initial net construction and the $\Rel(\cdot)$-construction require the following primitive
for a point set $Q$, which we call a \emph{near-neighbour query}:
\textit{Given a point $q\in Q$ and a radius $r$, return a list of points in $Q$ containing
exactly the points at distance $r$ or smaller from $q$).}
In the remainder of the section, we give more details next on how to compute net and the associated clusters, 
and how to find the neighbours for each such cluster, assuming that we have a primitive which 
can perform near-neighbour queries. 
In Section~\ref{sec:near-neighbour}, we show the implementation of such an primitive.

%%%%%

\paragraph{Net construction}
We construct the net using a greedy scheme: For any input point, store a pointer 
$\netpoint{p}$ pointing to the net point assigned to point $p$. 
Initially, $\netpoint{p}\gets NULL$ for all $p$. 
As long as there is a point $p$ with $\netpoint{p}=NULL$, we set $\netpoint{p}\gets p$
and query the near-neighbour primitive to get a list of points with distance at most $t$ from $p$. 
For any point $q$ in the list
we update $\netpoint{q}\gets p$ if either $\netpoint{q}=NULL$ or \mbox{$\distance{p}{q}<\distance{N(q)}{q}$}.
Then we pick the next point $p$ with $\netpoint{p}=NULL$.

At the end, the set of points $p$ with $\netpoint{p}=p$ represent the net at scale $\trep$ and points
$q$ satisfying $\netpoint{q}=p$ constitute $p's$ cluster. All points are assigned to their closest
net-point. The net thus constructed is a $(\trep,\trep)$-net. Moreover, we assign the same level to all root 
clusters. In particular, for any root node $v$, we set $\level(v)$ such that
$\frac{2\tau}{\tau-1}\cdot\tau^{\level(v)}=\trep$. 
Specifically, we set $\level(v)=\floor{\log_{\tau} \big(\frac{\tau-1}{2\tau}t\big)}$.

\paragraph{Computing the Rel(.)\ set for the roots}
After computing the net-points and their respective clusters, we need to augment the net-points with
neighbouring information. Recall that $\Rel(u)$ contains nodes in distance at most $14\tau^{\level(u)}$ from
$rep_u$. Since we have a $(\trep,\trep)$-net, the level of any root node $u$ satisfies
$14\tau^{\level(u)}\le 7\trep$. Hence we need to
find neighbours of net-points within $7\trep$, the minimum distance between any two net-points being more than $\trep$. 
By the doubling property, any root net-node can have at most $\lambda_{7\trep}^{\log_2 \frac{7\trep}{\trep/2}}$ 
such neighbours which simplifies to $C'=14^{\tds{7\trep}}$. We use the near-neighbour primitive to compute such neighbours.

\section{Near-neighbours primitive}
\label{sec:near-neighbour}

We describe the primitive used in the previous section which performs near-neighbour queries.
Our approach follows the notion of 
\emph{Locality-sensitive hashing} (LSH) introduced by \cite{old-lsh} for the Hamming metric 
and extended to Euclidean spaces in \cite{lsh}. LSH is a popular approach to find approximate
near-neighbours in high dimensions.

\paragraph{Locality Sensitive Hashing}

LSH applies several hash functions on a point set such that close points are more likely to map to the same hash-buckets 
than points which are sufficiently far away. A typical application of LSH is the $(r,c)$-\emph{nearest neighbour}
problem: If there exists a point within distance $r$ of the query point $q$, report some point within distance
$cr$ of $q$, $c>1$.

However, for our construction we wish to solve the following problem: report \emph{all} points within distance $r$ of 
the query point. We need the LSH oracle for two steps in our construction: constructing the net at scale $\trep$
and computing the $\Rel(\cdot)$ for the root-nodes.
We show that both these steps requires a runtime sub-quadratic in $n$ by a slight modification of the 
method presented in~\cite{lsh}.
We repeat some of their definitions for clarity:

\begin{defi}
 A family of hash functions $\mathcal{H}=\{h:S\rightarrow U\}$ is called 
\emph{$(r_1,r_2,p_1,p_2)$-sensitive} if for all $a,b \in S$, the following holds:
 \begin{itemize}
  \item if $\distance{a}{b}\le r_1, Pr_1=P[h(a)=h(b)]\ge p_1$
  \item if $\distance{a}{b}\ge r_2, Pr_2=P[h(a)=h(b)]\le p_2$
  \item $p_1 \ge p_2$ and $r_1 \le r_2$
 \end{itemize}
\end{defi}
We amplify the gap between $Pr_1$ and $Pr_2$ by concatenating $k$ such hash functions,
creating the family of hash functions $\mathcal{G}=\{g:S\rightarrow U^k\}$ such that 
\mbox{$g(x)=(h_1(x),h_2(x),...,h_k(x))$.} For $g(x)$, we have the modified properties:
  \begin{itemize}
    \item if $\distance{a}{b} \le r_1, P[g(a)=g(b)]\ge p_1^k$
    \item if $\distance{a}{b} \ge r_2, P[g(a)=g(b)]\le p_2^k$
  \end{itemize}
We describe our near-neighbour primitive next:
The input is a point set $Q$ with $n$ points and a distance $r>0$. As pre-processing step,  
we choose $l$ hash functions $g_1,....,g_l$ uniformly at random from $\mathcal{G}$ \cite[Sec.3]{lsh}. 
and hash each $p \in Q$ to the buckets $g_i(p)\forall i\in [1,l]$. 
Given a query point $q\in Q$, 
we iterate over $i=1,\ldots,l$ and check for any point $p$ in bucket $g_i(q)$ whether the distance to $q$ is at most $r$.
We output the points with this property as the near-neighbours of $q$.

We need to specify the parameters of LSH in the above description.
Most importantly, we have to ensure
that, with high probability, the output contains all points in distance $r$ from $q$.
Moreover, we want the buckets to be of small size so that the primitive does not have to filter out
too many false positives.

The performance of the LSH scheme depends upon a parameter $\rho\in(0,1)$ which appears as an exponent
of $n$ in the runtime.
We choose the parameters $p_1$, $p_2$, $r_1$ and $r_2$ of the hashing scheme such that
$\rho=\frac{\log p_1}{\log p_2}\approx \frac{r_1}{r_2}$~\cite[Sec.4]{lsh}. In the following parts
of the section, we let $\rho=\frac{r_1}{r_2}$.

\begin{lemma}
\label{lem:lsh-net-kl}
Let  $r_1:=r$ and $r_2:=r/\rho$,
$k:=\ceil{-\log_{p_2} n}$ and $l:=\ceil{2n^{\rho}\ln \frac{n}{\sqrt{\delta}}}$
with an arbitrarily small constant $\delta$.
The near-neighbour primitive has the following properties:
 \begin{enumerate}[label=(\roman{*}), ref=(\roman{*})]
 \item \label{net:cond1} With probability at least $1-\delta$, all points in distance at most $r$ are reported 
 for all query points.
 
 \item \label{net:cond2} For any query point $q$, the aggregate expected size of all buckets $g_1(q),...,g_l(q)$ is 
 at most $l(\wt{C}+1)$,
where $\wt{C}$ is the number of points in $Q$ with distance at most $r_2$ to $q$.

 \item \label{net:cond3} The pre-processing runtime is $O(dnkl)$ 
and the expected query runtime for a point is is $O(dl(k+\wt{C}))$, where $\wt{C}$ is defined as in~\ref{net:cond2}.
\end{enumerate}
\end{lemma}

\begin{proof}
First we bound the expected aggregate size of the buckets. 
A bucket contains ``close'' points which are in distance at most $r_2$ from $q$
and ``far'' points which are further away. However, since the probability of 
a far point falling in the same bucket as $q$ is at most $p_2^k$, the expected
size of a single bucket is at most $\wt{C}+np_2^k\leq \wt{C}+1$ by our choice of $k$.
Since there are $l$ buckets, \ref{net:cond2} is satisfied.

For~\ref{net:cond1}, fix two points $q_1,q_2\in Q$ with distance at most $r_1$.
We have to ensure that $g_j(q_1)=g_j(q_2)$ for some $j\in\{1,\ldots,l\}$; this implies
that $q_1$ will be reported for query point $q_2$, and vice versa.
The probability for $g_j(q_1)=g_j(q_2)$ for a fixed $j$ is at least $p_1^k$,
which is $p_1^{-\log_{p_2} n}=n^{-\rho}$.
Hence the probability that $g_j(q_1)\neq g_j(q_2)$ holds for all $j\in\{1,\ldots,l\}$ is at most
$(1-n^{-\rho})^l$ because we choose the hash functions uniformly at random. 
There are at most $n^2$ point pairs within distance at most $r_1$. By the union bound, the probability
that at least one such pair maps into different buckets is at most
$n^2(1-n^{-\rho})^l$. Now we can bound
\begin{eqnarray*}
n^2(1-n^{-\rho})^l&=&n^2(1-n^{-\rho})^{2n^{\rho}\ln \frac{n}{\sqrt{\delta}}}\\
&=&n^2(1-\frac{1}{n^{\rho}})^{n^{\rho}\ln \frac{n^2}{\delta}}\\
&\leq& n^2 e^{-\ln \frac{n^2}{\delta}}=\delta,
\end{eqnarray*}
where we used the fact that $(1-1/x)^x\leq 1/e$ for all $x\geq 1$.
It follows that the probability that all pairs of points in distance at most $r_1$
fall in at least one common bucket is at least $1-\delta$.
This implies~\ref{net:cond1}.

It remains to show~\ref{net:cond3}: in the pre-processing step, we have to compute $k\cdot l$ hash functions for $n$ points.
Computing the hash value for a point $p$, $h_i(p)$ takes $O(d)$ time \cite[Sec.3.2]{lsh}.
For a query, we have to identify the buckets to consider in $O(dkl)$ time
and then iterate through the (expected) $l(\wt{C}+1)$ candidates (using \ref{net:cond2}), 
spending $O(d)$ for each.
\end{proof}

\paragraph{Net-forest construction using LSH}
We analyze the complexity of our net-forest construction from Section~\ref{sec:net-forest}
with the near-neighbour primitive that uses LSH.
The primitive is used in the construction of the $(\trep,\trep)$-net,
where we find the near-neighbours of distance at most $\trep$ for a subset of points that form the net
in the end.
That means, we initialize the primitive with $r\gets \trep$ and $Q\gets P$.
\begin{lemma}
\label{lem:lsh-net-time}
 The expected time to construct the $(\trep,\trep)$-net using LSH is
$$O\bigg( dn^{1+\rho}\log n \Big( \log n + \left(\frac{2}{\rho}\right)^{\tds{\trep/\rho}} \Big) \bigg).$$
\end{lemma}

\begin{proof}
We consider the time spend on all near-neighbour queries: Let the resulting net
consist of $m\leq n$ points. 
This implies that the algorithm proceeds in $m$ rounds
and queries the near-neighbours of $m$ points. Let $\wt{C}_i$ denote the number
of points in distance $t/\rho$ from the $i$-th query point. 
By Lemma~\ref{lem:lsh-net-kl}, the total complexity for all near-neighbour queries is: 
\begin{eqnarray}
O\left(ndkl+\sum_{i=1}^m dl(k+\wt{C}_i)\right)=O(ndkl+dl\sum_{i=1}^m \wt{C}_i)
\label{eqn:net_complexity}
\end{eqnarray}
We only need to bound the sum of the $\wt{C}_i$. 
For that, we fix some $q\in P$ and count in how many sets $\wt{C}_i$ may it appear.
Let $p_i$ denote the net-point chosen in the $i$-th iteration. 
We call such a net-point \emph{close} to $q$ if the distance to $q$ is at most $\trep/\rho$.
By definition, the net-points close to $q$ lie in a ball of radius $\trep/\rho$ centered at $q$.
Since any pair of net-points has a distance of more than $\trep$, any ball
of radius $\trep/2$ can contain at most one close net point.
Following the definition of the $\trep$-restricted doubling dimension, the number of such net-points can be at most
$\lambda_{\trep/\rho}^{\log_2\frac{\frac{\trep}{\rho}}{\frac{\trep}{2}}}$ which simplifies to 
$\big(\frac{2}{\rho}\big)^{\tds{\trep/\rho}}$. It follows that 
$$\sum_{i=1}^m \wt{C}_i \leq n \big(\frac{2}{\rho}\big)^{\tds{\trep/\rho}}.$$
Plugging this into $\eqref{eqn:net_complexity}$, we get the claimed running time, observing that
$k=O(\log n)$ and $l=O(n^\rho\log n)$ by Lemma~\ref{lem:lsh-net-kl}.
All additional operations in the net construction besides the calls of the primitive are 
dominated by that complexity.
\end{proof}

The second appearance of the near-neighbour primitive is in the construction of the $\Rel(.)$ sets
for the roots of the net forest. Recall that the roots are represented by the net-points
constructed before; let $M$ denote the set of net-points and $m$ their cardinality.
We simply have to find all pairs of points of distance at most $7\trep$ among the net-points;
and to do so we call the near neighbour primitive with $r\gets 7\trep$ and $Q\gets M$ for all $q\in M$.
\begin{lemma}
\label{lem:lsh-rel-time}
Computing the $\Rel(.)$ sets using LSH takes expected time
$$O\left( dn^{1+\rho}\log n\left(\log n + \left(\frac{14}{\rho}\right)^{\tds{7\trep/\rho}}\right) \right).$$
\end{lemma}

\begin{proof}
The proof is analogous to the proof of Lemma~\ref{lem:lsh-net-time}:
Let $\wt{C}_i$ (for $i=1,\ldots,m$) denote the number of net-points in distance at most $\frac{7\trep}{\rho}$
to the $i$-th net point. The same packing argument as in the previous Lemma shows that any $\wt{C}_i$
can be at most $(\frac{14}{\rho})^{\tds{7\trep/\rho}}$, so the their sum is bounded by 
$m(\frac{14}{\rho})^{\tds{7\trep/\rho}}$.
Analogous to the proof of Lemma~\ref{lem:lsh-net-time}, we can thus bound the runtime to
be as required, noting that $m\leq n$.
\end{proof}

\begin{theorem}
  \label{thm:lsh-time}
The expected time for constructing the net-forest using LSH is
$$O\left( dn^{1+\rho}\log n\left(\log n + \left(\frac{14}{\rho}\right)^{\tds{7\trep/\rho}}\right) \right).$$
\end{theorem}

\begin{proof}
Using Lemma~\ref{lem:lsh-net-time} and Lemma~\ref{lem:lsh-rel-time}, constructing the net and its $\Rel(\cdot)$ sets
are within the complexity bound. Constructing a single net-tree for a node containing $n_i$ points takes time at most 
$2^{14\td}n_i\log n_i$ (the factor of $14$ in $2's$ exponent can be seen by a careful analysis of~\cite[Sec.3.4]{hm-fast}).
Constructing individual net-trees for the clusters takes time:
 $\sum _{i=1} ^{m} 2^{14\td}d n_i \log n_i$
 Since $\sum_{i=1}^{m}n_i=n$,
 the above runtime simplifies to $2^{14\td}d n \log n$. 
 Augmenting the net-forest takes time $dn2^{14\tds{7t}}$~\cite[Sec.3.4]{hm-fast}.
 The runtime for the latter steps are dominated by the $\Rel(\cdot)$ and net construction for 
 sufficiently large values of $n$.
\qedhere
\end{proof}
We see how the choice of $\rho$ affects the complexity bound: For $\rho$ very close to zero, we get a almost
linear complexity in $n$, to the price that we have to consider larger balls in our algorithm and thus increase
the restricted doubling dimension.

\section{Applications}
\label{sec:apps}

\paragraph{Well-Separated Pair Decomposition}
A pair of net-tree nodes $(u,v)$ is \emph{$\eps$-well-separated} if 
$\max\{\diam_u,\diam_v\}\le \eps \nodedist{u}{v}$, where $\nodedist{u}{v}$
denotes the distance between the representatives of $u$ and $v$. Informally speaking,
all pairs of points $(p,q)$ with $p\in P_u$, $q\in P_v$ have
a similar distance to each other if $(u,v)$ is well-separated.
A \emph{$\eps$-well-separated pair decomposition} ($\eps$-WSPD) is a collection of $\eps$-well-separated
pairs such that for any pair $(p,q)\in P\times P$, there exists a 
well-separated pair $(u,v)$ such that $p\in P_u$ and $q\in P_v$; we say that such a pair $(p,q)$ 
is \emph{covered} by $(u,v)$.

An $\eps$-WSPD of size $n\eps^{-O(\dd)}$ can be computed in time 
$d\left( 2^{O(\dd)}n \log n + n(1/\eps)^{O(\dd) } \right)$~\cite[Sec.5]{hm-fast}. 
A WSPD considers pairs over all scales of distance, just because it has to cover any pair of points.
We relax that condition and only require that all pairs of points in distance at most $\trep$ are covered.
We call the resulting structure \emph{$t$-restricted $\eps$-WSPD}.

We construct the $\trep$-restricted $\eps$-WSPD as follows: We start by constructing the corresponding augmented net forest;
let $u_1,\ldots,u_m$ be its roots. Since we know the $\Rel(\cdot)$ set for any root, we can identify
pairs $(u_i,u_j)$ such that $u_i$ is in $\Rel(u_j)$ and vice versa (this also includes pairs where $u_i=u_j$). 
For any such pair, we call \texttt{genWSPD}$(u_i,u_j)$
from \cite[Sec.5]{hm-fast}, which simply traverses the sub-trees until it finds well-separated pairs. 
We output the union of all pairs generated in this way.

\begin{theorem}
\label{thm:t-wspd}
For $0<\eps<1$ and $t>0$, our algorithm computes a $\trep$-restricted $\eps$-WSPD of size $n\eps^{-O(\tds{7\trep})}$
in expected time
$$NF + dn\eps^{-O(\tds{7t})}.$$
where $NF$ is the complexity for computing the net-forest from Theorem~\ref{thm:lsh-time}.
\end{theorem}
\begin{proof}
For correctness, any pair of nodes generated is $\eps$-well-separated by definition.
For the relaxed covering property, consider a pair $(p,q)$ of points in distance at most $t$. There are roots
$u_1$, $u_2$ in the net-forest with $p\in P_{u_1}$ and $q\in P_{u_2}$. Since the diameter of $u_1$ and $u_2$
is at most $2t$, the distance of $\rep_{u_1}$ and $\rep_{u_2}$ is at most $5t\leq 14\tau^{\level(u_i)}$.
Therefore, $u_2\in\Rel(u_1)$ (and vice versa), and there will be a pair generated that covers $(p,q)$. 

For the size bound, we can use the same charging argument as in~\cite[Sec.5]{hm-fast}. We can additionally
ensure by our construction that in all doubling arguments, the radius of the balls in question is at most $7t$
and therefore replace the doubling dimension by $\tds{7t}$ in the bound. The running time follows because
the number of recursive calls of \texttt{genWSPD} is proportional to the output size, and we spend $O(d)$
time per recursion step.
\end{proof}

\paragraph{Well Separated Simplical Decomposition}
The concept of well-separated simplical decomposition (WSSDs) of point sets, introduced by 
Kerber and Sharathkumar~\cite{kbr-cech} and extended to doubling spaces by Choudhary et al~\cite{cks-approximate},
generalizes the concept of WSPD to larger tuples. 
A $(\kapa+1)$-tuple $(v_0,v_1,\ldots,v_k)$ is called \emph{$\eps$-well separated} if 
each $v_i$ is a node of the net-tree and for any ball $\ball$ which contains at least one point
 of each $v_i$, it holds that
$$v_0\cup v_1\cup ....\cup v_k\subseteq(1+\eps) \ball$$
where $(1+\eps) \ball$ is the ball with same center as $\ball$ and radius multiplied by $(1+\eps)$. 
An $(\eps,\kapa)$-WSSD is a set of $\eps$-well-separated tuples of size up to $\kapa+1$
such that any $\kapa$-simplex is covered by some tuple.
In~\cite{cks-approximate}, an $(\eps,\kapa)$-WSSD of size $n(2/\eps)^{O(\dd \cdot \kapa)}$
is constructed in time $d\left( 2^{O(\dd)}n\log n+n(2/\eps)^{O(\dd \cdot \kapa)} \right)$.

Similar as before, we define a \emph{$\trep$-restricted $(\eps,\kapa)$-WSSD} to be a collection
of $\eps$-well-separated tuples such that each $\kapa$-simplex that fits into a ball of radius $t$
is covered by a tuple. The statement is equivalent to the condition that the radius
of the smallest minimum enclosing ball containing points from each node of the tuple is at most $t$.

\begin{theorem}
\label{thm:t-wssd-full}
 A \emph{$\trep$-restricted $(\eps,\kapa)$-WSSD} of size $n(\frac{2}{\eps})^{O(\tds{7\trep}\cdot\kapa)}$
 can be computed in time 
 $$NF+nd(\frac{2}{\eps})^{O(\tds{7\trep}\cdot\kapa)},$$
 where $NF$ is the complexity for computing the net-forest from Theorem~\ref{thm:lsh-time}.
 Within the same time bound, 
 we can construct a sequence of approximation complexes $(\Approx_\alpha)_{\alpha\in[0,\trep]}$
 of size $n(\frac{2}{\eps})^{O(\tds{7\trep}\cdot\kapa)}$ whose persistence module
 is an $(1+\eps)$-approximation (in the sense that the two modules are \emph{interleaved}~\cite{chazal})
of the \emph{truncated \Cech module} $(\mathcal{C}_\alpha)_{\alpha\in[0,\trep]}$.
 
\end{theorem}
We defer the description of the algorithm to construct the \emph{$\trep$-restricted $(\eps,\kapa)$-WSSD} 
and the proof of Theorem~\ref{thm:t-wssd-full} to Appendix~\ref{app:a}.

\paragraph{Approximating the $\trep$-doubling dimension} One can approximate $\td$ 
for any point set $P$ up to a constant factor by constructing a net-forest $\tree$ of scale $\trep$ over $P$.
Let $x$ denote the maximum out-degree of any node in $\tree$. 
Then $\log x$ is a constant approximation of $\td$. This follows from the arguments of~\cite[Sec.9]{hm-fast}.

\section{Conclusion and future work}
\label{sec:conclusion}
In this paper we presented an algorithm to construct a hierarchical net-forest up to a certain scale and
applied it to the construction of WSPDs and approximate \Cech complexes.
One possible optimization we have ignored in our analysis is that the packing
arguments we use are for the complete point set. However, during the $\Rel(\cdot)$ construction, we work with the 
net-points which satisfy certain packing properties. Since these constitute a subset of the original point set, they 
may have an even lower doubling dimensions which we could exploit. 
Finding more applications for the $\trep$-restricted doubling dimension is another 
direction which we would like to look into.

\paragraph{Acknowledgements}
This research is supported by the Max Planck Center for Visual Computing and Communication.

\appendix
\section{Proof of Theorem~\ref{thm:t-wssd-full}}
\label{app:a}

\paragraph{Construction of the $\trep$-restricted WSSD}
 We describe the algorithm to construct the $\trep$-restricted $(\eps,\kapa)$-WSSD
 and prove its correctness and runtime. In this appendix, we will heavily rely
 on the notations, algorithms, and results presented in~\cite{cks-approximate}.
 The algorithm proceeds iteratively; for $\kapa=1$, we construct a $(2\trep)$-restricted $\eps/2$-WSPD using the algorithm
 from Section~\ref{sec:apps}. 
 To construct $\Gamma_{\kapa+1}$ from $\Gamma_{\kapa}$, we iterate over the tuples $\gamma\in \Gamma_{\kapa}$.
 We use the scheme of~\cite[Sec.3]{cks-approximate}, computing an 
 approximate meb of $\gamma$ and then exploring ancestors of $v_0$ and their descendants at appropriate levels. 
 The only complication arises when the algorithm requests for an ancestor higher than root of the tree of $v_0$.
 In such a case, our algorithm
 uses the root as the ancestor. In the following lemma, we will show that with this approach, we still cover
 all simplices with meb radius of at most $t$.
 \begin{lemma}
  \label{lem:t-wssd-correct}
  The algorithm computes a $\trep$-restricted $(\eps,\kapa)$-WSSD.
 \end{lemma}
 \begin{proof}
We show by induction that with modified ancestor search, we still cover all simplices with meb radius at most $\trep$.
For $\kapa=1$, the correctness of the algorithm follows from Theorem~\ref{thm:t-wspd} in Section~\ref{sec:apps},
Let $\Gamma_{\kapa-1}$ cover all $(\kapa-1)$-simplices $\gamma$ which satisfy 
 $\mebrad(\gamma)\le \trep$. Consider any $\kapa$-simplex $\sigma=(m_0,\ldots,m_{\kapa})$ with $\mebrad(\sigma)\le \trep$.
 From~\cite[Lem.9]{cks-approximate}, there exists a point (say $m_{\kapa}$) such that $m_{\kapa}\in2\meb(\sigma')$
 where $\sigma':=\sigma\setminus \{m_{\kapa}\}$ and $2\meb(\sigma')$ represents a ball with twice the radius and the same
 center as $\meb(\sigma')$. Since $\sigma'$ is a $(\kapa-1)$-simplex and $\mebrad(\sigma')\leq\mebrad(\sigma)\leq\trep$, 
it is covered by some $\kapa$-tuple 
 $\gamma=(v_0,\ldots,v_{\kapa-1})\in\Gamma_{\kapa-1}$.
To prove correctness, we show that when our algorithm reaches tuple $\gamma$, it produces a $(\kapa+1)$-tuple $(\gamma,x)$ 
such that $m_{\kapa}\in P_x$
 which implies that the simplex $\sigma$ is covered by the $(\kapa+1)$-tuple $(\gamma,x)$.
 
When handling $\gamma$, the algorithm searches for an ancestor of $v_0$ at an appropriate scale. 
If this ancestor is found within the tree
of $v_0$ in the net-forest, the arguments from~\cite[Lem.12]{cks-approximate} carry over to ensure 
that a suitable $x$ is found.
So let us assume that the algorithm chooses the root of the tree of the net-forest that $v_0$ lies in.
Call that root node $a_0$.
The algorithm considers all nodes in $\Rel(a_0)$ and creates new tuples with their descendants.
Moreover, the net-forest contains a leaf representing the point $m_\kapa$; let $a'$ denote the root of its tree.
It suffices to show that $a'\in\Rel(a_0)$. Since $\mebrad(\sigma)\leq\trep$, the distance of $m_0$ and $m_k$ 
is at most $2\trep$.
Moreover, the distance of $m_0$ to $\rep_{a_0}$ is at most $t$, because the representatives of the roots form 
a $(t,t)$-net. The same holds
for $m_\kapa$ and $a'$. Using triangle inequality, the distance of $\rep_{a_0}$ and $\rep_{a'}$ is at most $4t$. 
This implies that $a'\in\Rel(a_0)$.
\end{proof}

\begin{lemma}
\label{lem:t-wssd-size}
 The size of the computed $\trep$-restricted $(\eps,\kapa)$-WSSD $\Gamma_{\kapa}$ is 
 $n(\frac{2}{\eps})^{O(\tds{7\trep}\cdot\kapa)}$.
\end{lemma}

\begin{proof}
The proof of~\cite[Lem.13]{cks-approximate} carries over directly~-- indeed, we can replace 
all occurrences of $\dd$ by $\tds{7\trep}$.
This comes from the fact that a node $u$ has at most $14^{\tds{7\trep}}$ nodes in $\Rel(u)$,
and for any node in $\Rel(u)$ we reach descendants of a  level of at most $O(\log(2/\eps))$ 
smaller then $u$ (see the proof of~\cite[Lem.13]{cks-approximate}
for details). Since every node in the net-forest has at most $2^{O(\tds{\trep})}$ children, we create at most
$$14^{\tds{7\trep}}(\frac{2}{\eps})^{O(\tds{\trep})}=(\frac{2}{\eps})^{O(\tds{7\trep})}$$
tuples in $\Gamma_{\kapa}$ from a tuple in $\Gamma_{\kapa-1}$. With that, the bound can be proved by induction.
\end{proof}

\begin{lemma}
\label{lem:t-wssd-time}
Computing a $\trep$-restricted $(\eps)$-WSSD takes expected time 
$$NF+nd(2/\eps)^{O(\tds{7\trep}\cdot\kapa)}$$
where $NF$ is the complexity for computing the net-forest from Theorem~\ref{thm:lsh-time}.
\end{lemma}

\begin{proof}
The proof is analogous to~\cite[Lem.14]{cks-approximate}, plugging in the 
running time for $t$-restricted $\eps$-WSPD from Theorem~\ref{thm:t-wssd-full} 
and the size bound from Lemma~\ref{lem:t-wssd-size}.
\end{proof}

\paragraph{Computing the approximate \Cech filtration}
We use the scheme of~\cite[Sec.4]{cks-approximate} to construct the $(1+\eps)$-approximate filtration on the
$\trep$-restricted WSSD. The original construction works without modification. 
Using the notation from~\cite[Sec.4]{cks-approximate}.,
for any WST $\sigma=(v_0,v_1,\ldots,v_{\kapa})$ with $\level(v_i)\le\hval$, 
we add $\sigma'=(\vcell(v_0,\hval),\vcell(v_1,\hval),\ldots,\vcell(v_{\kapa},\hval))$ to $\Approx_\alpha$ if 
$\mebrad(\sigma')\le\theta_{\disc}$.
The only potential problem with the $\trep$-restricted case is that such a $\vcell()$ might be a node higher than 
a root of the net-forest. This cannot happen, however, since $\hval$ is chosen such that
$$\frac{2\tau}{\tau-1}\tau^h\leq \frac{\eps}{7}\alpha.$$
Since $\alpha\leq \trep$ and $\eps\leq 1$, we have that 
$$\hval < \lfloor\log_\tau\frac{\tau-1}{2\tau}t\rfloor=\level(u)$$
for any root $u$ in the net-forest.

\begin{thebibliography}{10}

\bibitem{asd-dd}
Patrice Assouad.
\newblock Plongements {L}ipschitziens dans $\mathbb{R}^n$.
\newblock {\em Bulletin de la Societe Mathematique de France}, 111:429--448,
  1983.

\bibitem{bc-smaller}
M.~B\u{a}doiu and K.~Clarkson.
\newblock Smaller core-sets for balls.
\newblock In {\em Proc.\ 14th ACM-SIAM Symp.\ on Discr.\ Alg.}, pages 801--802,
  2003.

\bibitem{chazal}
Fr{\'e}d{\'e}ric Chazal, David Cohen-Steiner, Marc Glisse, Leonidas~J. Guibas,
  and Steve~Y. Oudot.
\newblock Proximity of persistence modules and their diagrams.
\newblock In {\em Proceedings of the Twenty-fifth Annual Symposium on
  Computational Geometry}, SCG '09, pages 237--246, 2009.

\bibitem{cks-approximate}
Aruni Choudhary, Michael Kerber, and R.~Sharathkumar.
\newblock Approximate {C}ech complexes in low and high dimensions.
\newblock \url{http://people.mpi-inf.mpg.de/~achoudha/Files/Papers/ApproximateCech.pdf}

\bibitem{lsh}
Mayur Datar, Nicole Immorlica, Piotr Indyk, and Vahab~S. Mirrokni.
\newblock Locality-sensitive hashing scheme based on p-stable distributions.
\newblock In {\em Proceedings of the Twentieth Annual Symposium on
  Computational Geometry}, SCG '04, pages 253--262, 2004.

\bibitem{eh-computational}
Herbert Edelsbrunner and John Harer.
\newblock {\em Computational Topology. An Introduction.}
\newblock American Mathematical Society, 2010.

\bibitem{gon}
Teofilo~F. Gonzalez.
\newblock Clustering to minimize the maximum intercluster distance.
\newblock {\em Theor. Comput. Sci.}, 38:293--306, 1985.

\bibitem{lee-dstar}
Lee-Ad Gottlieb and Robert Krauthgamer.
\newblock Proximity algorithms for nearly doubling spaces.
\newblock {\em SIAM J. Discrete Math.}, 27(4):1759--1769, 2013.

\bibitem{hm-fast}
Sariel Har-Peled and Manor Mendel.
\newblock Fast construction of nets in low dimensional metrics, and their
  applications.
\newblock In {\em SIAM J. Comput}, pages 150--158, 2005.

\bibitem{old-lsh}
Piotr Indyk and Rajeev Motwani.
\newblock Approximate nearest neighbors: Towards removing the curse of
  dimensionality.
\newblock In {\em Proceedings of the Thirtieth Annual ACM Symposium on Theory
  of Computing}, STOC '98, pages 604--613, 1998.

\bibitem{kbr-cech}
Michael Kerber and R.~Sharathkumar.
\newblock Approximate {C}ech complexes in low and high dimensions.
\newblock In {\em International Symposium on Algorithms and Computation}, pages
  666--676, 2013.

\bibitem{talwar}
Kunal Talwar.
\newblock Bypassing the embedding: Algorithms for low dimensional metrics.
\newblock In {\em Proceedings of the Thirty-sixth Annual ACM Symposium on
  Theory of Computing}, STOC '04, pages 281--290, 2004.

\end{thebibliography}
\end{document}